\title{Test Martingales for bounded random variables}
\author{Harrie Hendriks\\
Radboud University Nijmegen}
\date{}

\documentclass[12pt]{article}
\newif\ifappendix\appendixfalse\appendixtrue
\usepackage{graphicx}
\usepackage{amsmath}
\usepackage{amsfonts}
\usepackage{amssymb}
\usepackage[english]{babel}
\usepackage[gen]{eurosym}
\newtheorem{theorem}{Theorem}

\newtheorem{corollary}[theorem]{Corollary}

\newtheorem{lemma}[theorem]{Lemma}

\newtheorem{remark}{Remark}

\newenvironment{proof}[1][Proof]{\noindent\textbf{#1}\quad}{\ \rule{0.5em}{0.5em}}
\def\text{\mbox}
\def\R{{\mathbb R}}

\def\P{{\mathbb P}}

\def\E{{\bf E}}

\def\Var{{\rm Var}}

\def\KL{{\rm KL}}

\def\Alt{{\mathcal B}}
\begin{document}

\maketitle

\begin{abstract}
Given a positive random variable $X$, $X\ge0$ a.s., 
a null hypothesis $H_0:\E(X)\le\mu$ and a random sample of infinite size of $X$,
we construct test supermartingales for $H_0$, 
i.e. positive processes that are supermartingale if the null hypothesis
is satisfied.
We test hypothesis $H_0$ by testing the supermartingale hypothesis 
on a test supermartingale.
We construct test supermartingales that lead to tests with power 1. 
We derive confidence lower bounds.
For bounded random variables we extend the techniques to two-sided tests of $H_0:\E(X)=\mu$ and to the construction of confidence intervals.
\\
In financial auditing random sampling is proposed as one of the possible techniques to gather
enough evidence to justify rejection of the null hypothesis that there is a 'material' misstatement in a financial report.
The goal of our work is to provide a mathematical context 
that could represent such process
of gathering evidence by means of repeated random sampling, while ensuring an intended significance level.
\end{abstract}

\emph{Mathematics Subject Classification:} Primary 62L12; Secondary 60G42, 62G10, 62G15

\emph{Keywords:} Sequential hypothesis test, maximal lemma, hypothesis on mean, 
nonparametric test, first passage time, Wald's equation,  
confidence lower bound, confidence interval, audit sampling, acceptance sampling,
Hoeffding's inequality.

\section{Introduction}\label{Introduction}
We are inspired by Gr\"unwald \cite{PG} and Shafer et al \cite{SSVV}
who pointed out the relationship between 
sequential probability ratio tests (\cite{W}) and martingale theory.
In \cite{PG} the test supermartingale concept is explicitly announced
as a contribution to the current discussion about 
the interpretation of \emph{p-value} in scientific literature.
Further developments along this line can be found in \cite{GHK}.
In terms of the gambling metaphor the $p$-value is replaced with the fortune of a gambler who starts with unit fortune,
and plays a sequence of games of chance which are fair or on average loss-making
if the null hypothesis is satified.
The cited works mainly describe tests concerning the parameters in a parametrized family of 
probability distributions. 
We will describe tests concerning the expected value of a random variable,
under the only assumption that the random variable 
is non-negative and its expectation exists (and is finite).
We are interested in the performance of these tests depending on the expectation 
and the variance of the random variable.
\\
Waudby-Smith and Ramdas in \cite{W-SR} constructed confidence intervals and time-uniform
confidence sequences for the mean of a bounded random variable
using the same test supermartingale technique as ours developed in Section \ref{H0}.
They observed that Stark \cite{S}, based on an idea of Kaplan (see \cite{HK}), has developed this test supermartingale technique, as well
as its integrated version and one based on sampling without replacement.
\\
\\
The novelty of this paper is the study of the performance of the test supermartingales defined in Section \ref{H0} depending on the expecation and the variance of a random variable not satisfying the null hypothesis (Section \ref{Behavior})
and an application with financial audit sampling in mind (Section \ref{Bounded case}).
In this paper we hope to reach not only statisticians with a reasonable background in probability,
but also applied statisticians. 
That is why we will explain some notions from probability.
We will say that some event is \emph{almost sure}, or a.s., 
if its probability is 1 with respect to the relevant
probability distribution(s). 
The term random variable may be abbreviated to rv. 
A random variable $X$ is \emph{integrable} if its expected value $\E(X)$ exists and is finite
and it will be called \emph{positive} if $X\ge0$ a.s..
A sequence of rv's $\{X_k\}_{k=1}^\infty$, 
$X_1,X_2,X_3,\ldots$,
is a random sample or an iid (independent identically distributed) sample of $X$ 
if it is a collection of independent rv's 
such that each $X_k$ has the same probability distribution as $X$.
\\
\\
\textbf{Example:} In the context of financial auditing,
we have in mind that $X$ is defined on some population $\Omega$, 
say a finite set $\Omega=\{\omega_1,\ldots,\omega_L\}$, in the sense that given $\omega\in \Omega$ 
there is a well defined procedure to determine its value $X(\omega)\in\R$.
The auditor has to assure himself that $\Omega$ is well defined and that the procedure to determine
an $X$-value is practically feasible.
The randomness of $X$ is introduced by the auditor 
who has a procedure to select randomly an item $\omega$ in $\Omega$, 
such that any $\omega$ has probability $p(\omega)$ to be selected, 
where $p(\omega)\ge0$ and $\sum_{\omega\in\Omega} p(\omega)=1$.
One is interested in a characteristic of $X$ that can be interpreted as 
its expected value $\E(X)=\sum_{\omega\in\Omega}X(\omega)p(\omega)$ with respect to $p$.
\\
For example $\omega_1,\ldots,\omega_L$ are identifiers of items underlying a financial report.
For $\omega\in\Omega$ one has its book value $B(\omega)>0$, the audited value $A(\omega)$
and the so-called taint or tainting $T(\omega)=(B(\omega)-A(\omega))/B(\omega)$.
In this context we will assume that $0\le A(\omega)$ 
so that $T(\omega)\le1$,
and usually one assumes that also $A(\omega)\le B(\omega)$, so that $T(\omega)\ge0$.
The total book value is 
$B_{\hbox{\scriptsize tot}}=\sum_{\omega\in\Omega}B(\omega)$. 
The auditor is interested in the total misstatement 
$$\sum_{\omega\in\Omega}B(\omega)-A(\omega)=B_{\hbox{\scriptsize tot}}\sum_{\omega\in\Omega}T(\omega)B(\omega)/B_{\hbox{\scriptsize tot}}=B_{\hbox{\scriptsize tot}}\sum_{\omega\in\Omega}T(\omega)p(\omega)
=B_{\hbox{\scriptsize tot}}\E(T),$$
where $p(\omega)=B(\omega)/B_{\hbox{\scriptsize tot}}$ (Probability Proportional to Size) satisfies the properties of a probability density. 
Clearly we do not offer a way out of handling missing items in the report (or items with book value 0).
\\
The problem is that $\Omega$ is a large set, 
so that it is not practical to determine all $T$-values.
Often it will be sufficient to determine an upper bound for $\E(T)$ (or, equivalently, a lower bound for $\E(1-T)$), 
based on an iid sample of $T$.  
Sampling may be carried out as follows.
Given a number $z\in[0,1]$ one may associate to it that item $w(z)=\omega_\ell\in \Omega$ such that
$\sum_{i=1}^{\ell-1}p(\omega_i)<z\le\sum_{i=1}^{\ell}p(\omega_i)$.
A random sample 
$T_1,T_2,\ldots$ of $T$ can be constructed with the help of a random number generator 
yielding a random sample of
numbers $z_1,z_2,\ldots$, uniformly distributed in [0,1], by 
associating to it $T_1=T(w(z_1)), T_2=T(w(z_2)),\ldots$.
Notice that all $\omega\in\Omega$ will occur (almost surely) infinitely often
in the sequence $w(z_1),w(z_2),\ldots$.
\\
\\
A  process (in discrete time) $\{Y_k\}_{k=0}^\infty$ is a sequence of 
rv's $Y_0,Y_1,Y_2,\ldots$.
The index $k$ is referred to as time, and we will speak about time $k$.
The notion of a filtration $\{\mathcal F_k\}_{k=0}^\infty$ is used to formalize the development in time of the state of the investigator.
More explicitly $\mathcal F_k$ is a $\sigma$-algebra which represents all the information that is available 
up to and including time $k$.
The process $\{Y_k\}_{k=0}^\infty$ is adapted to $\{\mathcal F_k\}_{k=0}^\infty$
if the variables $Y_0,Y_1,\ldots,Y_k$ are $\mathcal F_k$-measurable, that is, their values are measurable 
at time $k$.
The process is \emph{integrable} (resp. \emph{positive}) if each rv $Y_k$ is integrable (resp. positive).
Assume the integrable process $\{Y_k\}_{k=0}^\infty$ is adapted to $\{\mathcal F_k\}_{k=0}^\infty$.
Given $\ell$ and $k$ the conditional expectation $\E(Y_\ell\mid\mathcal F_k)$
is a rv that is $\mathcal F_k$-measurable. 
For each possible realization of observable values at time $k$
is associated a value of $\E(Y_\ell\mid\mathcal F_k)$, 
representing the expected value of $Y_\ell$, given that realization.
It holds that
$\E(Y_\ell\mid\mathcal F_k)=Y_\ell$ for $\ell\le k$.
The conditional expectation of rv $Y_\ell$ with respect to the trivial $\sigma$-algebra (no information)
corresponds to the ordinary notion of expectation.
Thus $\E(Y_\ell)=\E(\E(Y_\ell\mid\mathcal F_k))$.
\\
Given $0\le\nu\le1$, by $\Alt(\nu)$  we
denote the Bernoulli distribution with values 0 and 1 and expectation $\nu$, 
in particular the probability of 1 (resp. 0) is $\nu$ (resp. $(1-\nu)$).
\\
\\
In section \ref{Test martingales} we state the so-called maximal lemma 
and show how it leads to a
test that a random process is a supermartingale.
Suppose given a null hypothesis $H_0:\E(X)\le\mu$ about an integrable rv $X$ such that $X\ge0$ a.s..
In section \ref{H0}, 
given a
random sample $\{X_k\}_{k=1}^\infty$ of $X$, we develop a method to construct a process $\{M_k\}_{k=0}^\infty$
which is a positive supermartingale if $X$ satisfies $H_0$.
Such a process is called a test supermartingale for $H_0$.
We indicate how to handle samples without replacement and stratified samples. 
In section \ref{SPRT} we give a few examples of alternative test supermartingales.
Moreover we give a supermartingale approach to a theorem of Hoeffding.
In section \ref{Behavior} we study the behavior of test supermartingales, 
as depending on $X$ not satisfying the null hypothesis, 
that is $\E(X)>\mu$.
In section \ref{Bounded case} we go in detail to the case relevant in financial auditing,
as expounded on above,
and indicate its relation to current practice.
In section \ref{conf reg} we apply the technique to the construction of confidence upper bounds
and confidence intervals.

\section{Supermartingales}\label{Test martingales}

Suppose given an integrable process $\{Y_k\}_{k=0}^\infty$, 
adapted to the filtration $\{\mathcal F_k\}_k$.
Recall the $\sigma$-algebra $\mathcal F_k$ represents the information available at time $k$, including
the values of $Y_i$ for $i\le k$.
The process $\{Y_k\}_k$ is a \emph{supermartingale} 
if $Y_k\ge\E[Y_\ell\mid \mathcal F_k]$ for $0\le k < \ell<\infty$.
It is called a \emph{martingale} if $Y_k=\E[Y_\ell\mid \mathcal F_k]$ for $0\le k < \ell<\infty$.
%
%
Suppose $\{Y_k\}_k$ is a positive supermartingale.
The 'maximal lemma' (\cite[Ch. V, Thm. 20]{DM}, cf. Ville's gambler's ruin theorem \cite[Thm. 1 p.84]{Ville})
implies that
$$
\forall\lambda\ge0:\lambda\,\P\{\sup_kY_k\ge\lambda\}\le\E[Y_0].
$$
We apply this result in the following practical, but actually equivalent, form: 

\begin{lemma}\label{maximal lemma}
If $\{Y_k\}_k$ is a positive supermartingale,
 then
$$\forall\lambda\ge0:\lambda\,\P\{\exists \ell: Y_\ell\ge\lambda\}\le\E[Y_0].$$
\end{lemma}

\begin{proof}
Consider the random variable $N$ which is 
the first time $k$ that the process $\{Y_k\}_k$ reaches or exceeds level $\lambda$,
or $\infty$ if the process does not exceed level $\lambda$.
$N$ is a stopping time.
If one stops the supermartingale $\{Y_k\}_k$ at that time, 
the stopped process $\{Y_{k\wedge N}\}_k$ is still a supermartingale (\cite[Thm. 5.2.6]{D}).
Thus for any $\ell$  we have $\E[Y_0]=\E[Y_{0\wedge N}]\ge\E[Y_{\ell\wedge N}]$ while
for the positive random variable $Y_{\ell\wedge N}$ we have
$\E[Y_{\ell\wedge N}]\ge\lambda\P\{Y_{\ell\wedge N}\ge\lambda\}$. 
The lemma then follows since the events 
$\{Y_{\ell\wedge N}\ge\lambda\}=\{N\le\ell\}$ form an increasing sequence for increasing $\ell$
whose union is $\{N<\infty\}=\{\exists \ell:Y_\ell\ge\lambda\}$.
\end{proof}
\\
\\
\noindent
We follow \cite{PG} and \cite{SSVV} where 
the significance of the above ideas 
for
statistical hypothesis testing is worked out. 
Be given a statistical hypothesis $H_0$.
A \emph{test (super)martingale} (for $H_0$) is a process $\{Y_k\}_{k=0}^\infty$
such that, if $H_0$ is satisfied, 
$\{Y_k\}_k$ is a positive (super)martingale and $\E[Y_0]>0$.
Be given a significance level $\alpha$, $0<\alpha\le1$.
A practical test consists of observing sequentially $Y_k$, $k=1,2,\ldots$
and stop at time $n$ if $Y_n\ge\E(Y_0)/\alpha$ or stop at some other time $n$.
In the first case $Y_n\ge\E(Y_0)/\alpha$ and one may reject $H_0$, otherwise one cannot reject $H_0$.
The size of such test is the probability to reject $H_0$ under the assumption that $H_0$ is satisfied.
By Lemma \ref{maximal lemma} it is at most $\alpha$:
$$
\P\{\hbox{Reject }H_0\}\le\P\{\exists n:Y_n\ge\E(Y_0)/\alpha\}
\le(\E(Y_0)/\alpha)^{-1}\E(Y_0)=\alpha.
$$
To show the sharpness of Lemma \ref{maximal lemma}
we give the example of a classical test about the probability of success of a Bernoulli variable.
The basic idea behind the example is the following.
Let $H_0$ be a simple
null hypothesis about a rv $Z$ 
and let $R$ be the critical region of
a test of $H_0$ based on a random sample $Z_1,Z_2,\ldots$ of $Z$.
Then 
the conditional probabilities
$Y_i=\P(R\mid Z_1,\ldots,Z_i)$ of $R$ under hypothesis $H_0$, constitute a test martingale.

\begin{remark}\label{sharpness1/2}
Let $Z$ be a Bernoulli variable $Z\sim\Alt(\nu)$, $\nu$ unknown,
and consider null hypothesis $H_0:\E(Z)\ge\mu$,
that we want to test against $\E(Z)<\mu$.
Take the test to reject $H_0$ if in a sample of fixed size $n$ at most $k$ successes are found.
It will have significance level $\alpha=F(k;n,\mu)$, 
where $F(k;n,\mu)$ denotes the cumulative distribution function 
of the binomial distribution with parameters $(n,\mu)$ evaluated at $k$.
\\
We construct a test supermartingale $\{Y_i\}_i$ as follows.
Given a random sample  $Z_1,Z_2,\ldots$ of $Z$, let $S_i=Z_1+\cdots+Z_i$
%
and let 
$$
Y_0=F(k;n,\mu);\quad 
Y_i=F(k-S_i;n-i,\mu)\quad(1\le i\le n);\quad 
Y_i=Y_n\quad(i>n).$$
If 
$\nu=\E(Z)\ge\mu$, 
this process is a supermartingale because of a well-known recursion formula for
binomial distribution functions:  
\begin{align*}
\E(Y_i\mid S_{i-1})
&=\E(F(k-S_i;n-i,\mu)\mid S_{i-1})
\\&=
\nu F(k-(S_{i-1}+1);n-i,\mu)+(1-\nu)F(k-S_{i-1};n-i,\mu)
\\&\le
\mu F(k-(S_{i-1}+1);n-i,\mu)+(1-\mu)F(k-S_{i-1};n-i,\mu)
\\&=F(k-S_{i-i};n-(i-1),\mu)=Y_{i-1}.
\end{align*}
Notice that $[S_n\le k]\Leftrightarrow [Y_n=1]$ and $[S_n> k]\Leftrightarrow [Y_n=0]$.
If $\nu=\mu$ then $H_0$ holds and 
$\alpha=F(k;n,\mu)=\P\{Y_n=1\}\le\P\{\exists\ell:Y_\ell\ge1\}\le\E(Y_0)=\alpha$.
\end{remark}

\section{Test supermartingales for positive random variables}\label{H0}

Given $\mu>0$, we will construct test supermartingales to test null hypothesis $\E(X)\le\mu$
for a positive integrable rv $X$.
First we find functions $f:\R\to\R$ satisfying
\begin{equation}\label{testcond}
[X\ge 0\hbox{ a.s. \& }\E(X)\le\mu]\Rightarrow [f(X)\ge0\hbox{ a.s. \& }\E(f(X))\le 1].
\end{equation}
Applied to the case of two-point distributions 
let $0\le x<\mu<y$ and consider $X$ such that 
$\P\{X=x\}=(y-\mu)/(y-x)$, $\P\{X=y\}=(\mu-x)/(y-x)$, so that
$\E(X)=\mu$.
We need $f(x)(y-\mu)/(y-x)+f(y)(\mu-x)/(y-x)\le1$, or equivalently
$(1-f(x))/(\mu-x)\ge (f(y)-1)/(y-\mu)$.
This must hold for all $x,y$ with $0\le x<\mu<y$. 
Thus there is $b$ such that
$(1-f(x))/(\mu-x)\ge b\ge (f(y)-1)/(y-\mu)$ for all $x,y$ with $0\le x<\mu<y$.
It follows that $f(t)\le 1+b(t-\mu)=(1-b\mu)+bt$ for all $t\ge0$.
In order that $f(t)\ge0$ for all $t\ge0$ we need
$ 1-b\mu\ge0$ and $b\ge0$, that is $0\le b\mu\le1$.
%
Thus we propose the following functions: for $0\le c\le 1$
\begin{equation}\label{factor}
f(t)=(1-c)+c\cdot\frac{t}{\mu}.
\end{equation}
We would like to draw attention to its relation to likelihood ratios.
Suppose $0<c<1$, take $y>\mu$ and choose $\theta$ such that 
$c=(\theta-\mu)/(y-\mu)$.
Consider the two-point distributions $\ell_\mu$ and $\ell_\theta$ with support $\{0,y\}$ 
defined as
$\ell_\nu(y)=\nu/y$, $\ell_\nu(0)=1-\nu/y$ for $\nu=\mu$ resp. $\nu=\theta$.
Then $f(y)=(1-c)+c\,y/\mu=\theta/\mu=\ell_\theta(y)/\ell_\mu(y)$
and
$f(0)=1-c=(y-\theta)/(y-\mu)=\ell_\theta(0)/\ell_\mu(0)$.
Thus $f$ is the affine extrapolation of the likelihood ratio
$\ell_\theta/\ell_\mu$ at the points 0 and $y$.
\\
\\
See Section \ref{SPRT} for some alternatives to the condition $X\ge0$ a.s. in (\ref{testcond}).
%
\\
\\
Suppose $X$ is an integrable random variable 
such that $X\ge0$ a.s.. 
Consider the null hypothesis
$$H_0:\E(X)\le\mu.$$
We will construct 
test supermartingales for $H_0$ using factors modelled after (\ref{factor}).
Consider a random sample $X_1,X_2,\ldots$ of the random variable $X$.
It defines a filtration $\{\mathcal F_k\}_k$ by $\sigma$-algebras 
$\mathcal F_k=\sigma(X_1,\ldots,X_k)$ for $k\ge1$ and 
the trivial $\sigma$-algebra $\mathcal F_0$.
In particular $\E(X_k\mid\mathcal F_{k-1})=\E(X_k)=\E(X)$.
We let $M_0=1$. 
At time $(k-1)$ 
the variables $X_1,\ldots,X_{k-1}$ and $M_{k-1}$ are observed, having values $x_1,\ldots,x_{k-1}$ and $m_{k-1}$.
One chooses \emph{test parameter} $c_{k}\in[0,1]$ which should not depend on $X_k,X_{k+1},\ldots$ in any conceivable way
(this makes the process of test parameters $\{c_{k}\}_k$ `predictable', i.e. for all $k$, $c_{k}$ is $\mathcal F_{k-1}$-measurable).
and 
one defines 
\begin{equation}\label{test martingale}
M_k=m_{k-1}\cdot\left((1-c_k)+c_k\frac{X_k}{\mu}\right).
\end{equation}
For our choice of filtration the functional dependence of $c_{k}$ on the observations of $X_1,\ldots,X_{k-1}$ 
should have been
fixed before any observation was available.
But see Remark \ref{filtration} for a broader, more practical class of filtrations.
The process $\{M_k\}_k$ is a test supermartingale for $H_0:\E(X)\le\mu$.
We obtain a test with significance level $\alpha$ if we reject $H_0$ at a time $n$ with 
$\max\{M_k\mid k\le n\}\ge1/\alpha$.
\\
\\
As in \cite{PG} (and \cite{W-SR}) we will express the above construction in a \emph{gambling metaphor}, 
that we present as a
'martingale transform' of a supermartingale by a positive predictable process (see e.g. \cite[Thm. 5.2.5]{D}, \cite[Section 10.6]{Wpwm}).
We consider the hypothesis $H_0:\E(X)\le\mu$ that we would like to reject.
Based on the random sample $X_1,X_2,\ldots$ of $X$, 
consider the process $Y=\{Y_k\}_k$ with $Y_0=0$ and 
$Y_k=Y_{k-1}-1+X_k/\mu$.
Under $H_0$ the process $Y$ is a supermartingale 
with respect to the filtration $\{\mathcal F_k=\sigma(X_1,\ldots,X_k)\}_k$.
Consider a lottery that takes place at time $k$ and pays out $X_k/\mu$
per unit stake, so that the net gain per unit stake is 
$-1+X_k/\mu=Y_k-Y_{k-1}$. 
We start with an initial unit amount of capital $m_0=1$.
At time $k-1$ we have accumulated a capital of $m_{k-1}$ and
we decide to stake an amount of $H_k=c_{k}m_{k-1}$ in this lottery, $0\le H_k\le m_{k-1}$.
Then at time $k$ our capital will become 
$m_{k-1}+H_k(Y_k-Y_{k-1})=m_{k-1}+c_km_{k-1}(-1+X_k/\mu)=M_k$.
If $X$ satisfies hypothesis $H_0$, we have a fair or loss-making game.
In particular, if we succeed in ending up with a large gain, we have reason to assert that $\E(X)>\mu$.
With this metaphor it should be intuitively clear, that 
it is wrong to change the stake amount $c_{k}m_{k-1}$, after having observed $X_k$.
Moreover, if one wants to switch over to a new gambling game to reject $H_0$, 
one has to continue with the capital left after the preceding game.
%
\subsubsection*{Choosing the test parameters, integrated test supermartingales}
A plausible choice for $c_{k}$ in the construction 
(\ref{test martingale}) of a test supermartingale
is that value of $c$ that maximizes
$m_{k-1}(c)=\prod_{i=1}^{k-1}(1-c+c\,x_i/\mu)$, where some prudence is necessary to avoid $c_{k}=1$.
We will consider a different proposal.
Start with some probability density $\pi$ on $[0,1]$, 
typically the uniform probability density on the interval $[c_0,1]$ for some $c_0\ge0$.
Define $\{M_{k}(c)\}_k$ to be the test supermartingale based on the choice $c_{k}=c$, all $k$,
and consider the \emph{integrated test supermartingale} with respect to \emph{test measure} $\pi$:
\begin{align}
M_k(\pi)&=\int_0^1M_k(c)\pi(c)dc=\int_0^1M_{k-1}(c)\cdot(1-c+c\,X_k/\mu)\pi(c)dc
\label{its}
\\&=M_{k-1}(\pi)\cdot(1-c_{k}+c_{k}X_k/\mu), \hbox{ where } \nonumber
\\
c_{k}&=\int_0^1 c\cdot M_{k-1}(c)\pi(c)dc/M_{k-1}(\pi).  \nonumber 
\end{align}
Notice that $c_{k}$ is the expectation 
of the probability density $p_{k-1}$ defined by $p_{k-1}(c)=M_{k-1}(c)\pi(c)/M_{k-1}(\pi)$, 
and that $M_{k-1}(c)$ is a log-concave function in $c$ (see Remark \ref{L1}).
In case $\pi$ is the uniform probability distribution on [0,1],
for large $k$, density $p
_{k-1}$ will be concentrated around the value of $c$ for which
$M_{k-1}(c)$ is largest.
\\
In a numerical implementation one may choose a finite number of points
$b_1,\ldots,b_r\in[0,1]$ and assign to them probabilities $\pi(b_i)=\pi_i$ where
$\pi_i\ge0$ and $\pi_1+\cdots+\pi_r=1$ and consider
$M_k(\pi)=\sum_{i=1}^r M_k(b_i)\pi(b_i)$.
\\
See \cite{W-SR} where the above ideas are worked out in much detail in the context of
constructing confidence bounds.

\subsubsection*{Test supermartingales for the null hypothesis $H_0:\E(X)\ge\mu$ or $H_0:\E(X)=\mu$}
Given $X\le\tau$ a.s. and $\mu<\tau$, we obtain test supermartingales for the
null hypothesis $H_0:\E(X)\ge\mu$ by transforming it into $H_0:\E(\tau-X)\le\tau-\mu$
leading to multiplication factors
\begin{equation}\label{lower bound}
1-c_{k}+c_{k}\,\frac{\tau-X_k}{\tau-\mu},~k=1,2,\ldots,
\end{equation}
where $c_k$ is $\mathcal F_{k-1}$ measurable and $0\le c_{k}\le1$.
\\
\\
If $0\le X\le\tau$ a.s., one can combine test supermartingales $\{M_n^+\}_n$ for $H_0:\E(X)\ge\mu$ and 
 $\{M_n^-\}_n$ for $H_0:\E(X)\le\mu$, based on the same data,
by taking $\{\rho^+ M_n^++\rho^-M_n^-\}_n$ for any $\rho^+,\rho^-\ge0$ with $\rho^++\rho^-=1$.
This will be a test supermartingale for $H_0:\E(X)=\mu$.

\subsubsection*{Filtrations}
\begin{remark}\label{filtration}
In practice the observation of a rv $X_k$ is accompanied by some, possibly random, attributes
like the time and the monetary cost needed to determine the value of $X_k$.
In the financial auditing example as described in the Introduction one could also think of the 
book item 
and the book and audit value associated with the observation of the tainting.
In particular the actual filtration that one would like to adopt is much richer than
$\{\mathcal F_k^o=\sigma(X_1,\ldots, X_k)\}_k$, 
and may include these attributes,
as well as for example the mental condition of the investigator.
\end{remark}
In order to stay close to the intuition for a random, iid, sample $X_1,X_2,\ldots$, 
a suitable extra condition
on the sample is that $X_k,X_{k+1},\ldots$ and their attributes are independent 
of all information contained in $\mathcal F_{k-1}$.
One can reach this by actually \emph{hiding} previous to time $k$ all information about the rv's $X_\ell$ and their attributes for $\ell\ge k$
until it is decided to determine and process the value of $X_k$. 
On the other hand, a richer filtration might typically allow
test parameters $c_{k}$, depending not only on the observed values of $X_1,\ldots,X_{k-1}$,
but for example also on the built-up insights of the investigator up to time $k-1$.
\\
Inspired by \cite{W-SR}, as an illustration we present the example of 
\emph{sampling without replacement}.
We consider the situation of a filtration 
$\{\mathcal F_k\}_{k=0}^\infty$
and a sequence of random variables $\{X_k\}_{k=0}^\infty$ adapted to this filtration, such that the null hypothesis $\E(X)\le\mu$ is equivalent to the sequence of hypotheses
$\E(X_k\mid\mathcal F_{k-1})\le m_k$ where $m_k$ can be determined at time $k-1$,
based on the information available in $\mathcal F_{k-1}$.
In the financial auditing example described in the Introduction one could consider sampling without replacement, items $\omega$ with probability proportional to the book value $B(\omega)$, yielding a random ordering $(\omega_1,\ldots,\omega_L)$ of $\Omega$.
Null hypothesis $\E(T)\ge\mu$ is then equivalent to
$\E(T(\omega_k)\mid\mathcal F_{k-1})\ge m_k$ for $k=1,\ldots,L$ with
$$
m_k
=\frac{\mu B_{\hbox{tot}}-\sum_{i=1}^{k-1}T(\omega_i) B(\omega_i)}
  {B_{\hbox{tot}}-\sum_{i=1}^{k-1}B(\omega_i)} .
$$
If $m_k<0$ at some time $k$, one knows that $H_0$ is satisfied and sampling can be stopped.
One inductively defines a test supermartingale $\{M_k\}_k$ by $M_0=1$ and
$$M_k=M_{k-1}f_k(T_k)\hbox{ with } f_k(t)=(1-c_k)+c_k\cdot\frac{1-t}{1-m_k}$$
where $m_k$ is determined as before and $c_k$ is chosen at time $k-1$.
%
\\
\\
To conclude we sketch a way to handle \emph{stratified sampling}.
Suppose $\Omega_1,\ldots,\Omega_r$ is a partition of sample space $\Omega$.
Let $p_i=\P(\Omega_i)>0$ be known. 
Let $X:\Omega\to\R$ be a random variable such that $X\ge0$ a.s. and
denote by $\nu_i$ the conditional expectation $\nu_i=\E(X\mid\Omega_i)$
so that $\nu=\E(X)=\sum_i \nu_ip_i$.
We wish to test $H_0:\E(X)\le\mu$.
\\
Proposal:
For $m:\{1,\ldots,r\}\to(0,\infty)$, let $M_0^m=1$.
At time $k-1$, when $\omega_1,\ldots,\omega_{k-1}$ have been sampled and $M_{k-1}^m$ is determined, 
one decides from which stratum $\Omega_{j_k}$ to
randomly sample $\omega_k$.
Also at time $k-1$ one chooses $c_k$, such that $0\le c_k<1$.
One then defines
$$M_k^{m}=M_{k-1}^m(1-c_k+c_kX(\omega_k)/m(j_k)).$$
Reject $H_0$ if at some time $k$
it holds that
$M_k^m\ge1/\alpha$
 for all $m$ such that $\sum_im(i)p_i\le\mu$.
For $m$ defined by $m(i)=\nu_i$,
 $\{M_k^m\}_k$ is a martingale so that it exceeds $1/\alpha$ 
with probability at most $\alpha$.
In particular, if $\E(X)=\nu\le\mu$ the probability to reject $H_0$ is at most $\alpha$, 
so that the significance level of the test is at most $\alpha$.
\\
Suppose in each stratum $\Omega_j$
at some time $i\le k$, $\omega_i\in\Omega_j$ has been sampled such that $c_iX(\omega_i)>0$.
Then $\ell(m)=\log(M_k^m)$ has a minimum in the region $\sum_im(i)p_i\le\mu$.
Since $\ell(m)$ is decreasing in each $m(i)$, the minimum lies in the region $\sum_im(i)p_i=\mu$, which is convex.
Since $\ell(m)$ is strictly convex in $m$ this minimum is unique.
\\
It remains an open question how to choose at time $k-1$ stratum $j_k$ and test parameters $c_k$ depending on the ideas 
the investigator has, for example on the conditional expectations $\nu_i$ 
and the conditional variances, or how to optimize these choices based on the sample
$\omega_1,\ldots,\omega_{k-1}$ up to time $k-1$.

\section{Alternative test supermartingales}\label{SPRT}
In this section we present some variants of test supermartingales based on bounds
on the moment generation function of the random variable, 
rather than its positiveness.
Further elaboration of the proposals in this section can be found in \cite{W-SR}.

\begin{remark}\label{subGauss}
Consider the random variable $X$ and let $Z=X-\E(X)$.
Suppose random variable $Z$ is sub-Gaussian
 (see \cite{K}),
meaning that there exists $\tau>0$ such that
$$
\E(\exp(hZ))\le\exp\left(\frac12h^2\tau^2\right)\hbox{ for all }h\in\R.
$$
Suppose $Z$ is sub-Gaussian.
The Gauss deviation $\tau(Z)$ of $Z$ is the minimum $\tau$ for which this inequality holds.
It follows that all moments of $Z$ and $X$ exist and are finite, and that $\Var(X)=\Var(Z)\le\tau(X)^2$.
\\
A normal variable $X$ with standard deviation  $\sigma$
has the property that $Z=X-\E(X)$
is sub-Gaussian with Gauss deviation $\sigma$.
\\
Hoeffding's Lemma \cite[Inequality (4.16)]{H} implies that a random variable $X$ with values in 
an interval $[a,b]$ has the property that $Z=X-\E(X)$
is sub-Gaussian with Gauss deviation at most $\frac12(b-a)$.
\\
For each $h\ge0$ and $\mu$, function
$$f(t)=\frac{\exp(ht)}{\exp(h\mu+\frac12h^2\tau^2)}$$
satisfies the condition:
If $X-\E(X)$ is sub-Gaussian with Gauss deviation at most $\tau$ and 
$\E(X)\le\mu$ then
$f(X)\ge0$ and $\E(f(X))\le1$.
\end{remark}
Test supermartingale factor $f(t)$ as defined in Remark \ref{subGauss} 
corresponds to a likelihood ratio, 
namely of two normal distributions with standard deviation $\tau$:
with $\varphi(t)=(2\pi)^{-1/2}\exp(-t^2/2)$ and  $h=(\nu-\mu)/\tau^2$ we have
$$
f(t)=\frac{\exp(ht)}{\exp(h\mu+\frac12h^2\tau^2)}=\frac{\varphi((t-\nu)/\tau)}{\varphi((t-\mu)/\tau)}.
$$
Its inverse is the point of departure for a likelihood ratio ratio test of the null hypothesis
that $X\sim N(\mu,\tau^2)$ versus alternative $X\sim N(\nu,\tau^2)$.
\\
\\
For rv's with values in $[0,1]$, 
inspired by likelihood ratio tests for Bernoulli distributions, 
one may consider for $0<\mu<1$, $0<\nu<1$
$$
g(t,\nu,\mu)
=
\frac{\nu^t(1-\nu)^{1-t}}{\mu^t(1-\mu)^{1-t}}=
\frac{\exp(ht)}{(1-\mu)+\exp(h)\mu}\hbox{ with }h=\log\left(\frac{\nu(1-\mu)}{\mu(1-\nu)}\right).
$$
One can show that $\E(g(X,\nu,\mu))\le1$ if $0\le X\le1$ a.s. such that $(\E(X)-\mu)(\nu-\mu)\le0$.
Here the corresponding test supermartingale can be expressed in terms of 
densities of Beta distributions. 
Taking $H_0:\E(X)\le\mu$ and $\nu>\mu$,
notice that approach (\ref{factor}) with $c=(\nu-\mu)/(1-\mu)$ corresponds to affine interpolation and
improves on $g(t,\nu,\mu)$ which  is convex in $t$ so that
$$\frac{\nu^t(1-\nu)^{1-t}}{\mu^t(1-\mu)^{1-t}}\le \frac{1-\nu}{1-\mu}(1-t)+\frac\nu\mu t
=(1-c)+c\frac t\mu,
\quad\hbox{ for }0\le t\le1.$$
We make a side-step to a probability upper bound of Hoeffding (\cite{H}).
We reformulate his Theorem 1, Ineq. (2.1), as follows and recast his proof 
in terms of a supermartingale.
\begin{theorem}\label{Hoeffding}
Suppose $X_1,X_2,\ldots,X_n$ iid sample of $X$ with values in $[0,1]$ and $\E(X)=\mu$, $0<\mu<1$.
Let $\overline X=(X_1+\cdots+X_n)$ and $0<\mu<\theta<1$.
Then
$$
\P\{\overline X\ge\theta\}
\le
\left\{
\left(\frac\theta\mu\right)^\theta
\left(\frac{1-\theta}{1-\mu}\right)^{1-\theta}
\right\}^{-n}.
$$
\end{theorem}
\begin{proof}
For $h\ge0$, define the supermartingale $\{M_k(h)\}_k$
under $H_0:\E(X)\le\mu$:
$$M_k(h)=\prod_{i=1}^k\frac{\exp(hX_i)}{(1-\mu)+\mu\exp(h)}=\frac{\exp(h(X_1+\cdots+X_k))}{((1-\mu)+\mu\exp(h))^{k}}.$$
Let $h=\log(\theta(1-\mu)/(\mu(1-\theta)))$. Then $h>0$.
The event that at time $n$ we have $\overline X\ge\theta$ equals the event that
$$
M_n(h)
\ge
\frac{\exp(nh\theta)}{((1-\mu)+\mu\exp(h))^{n}}=\left\{
\left(\frac\theta\mu\right)^\theta
\left(\frac{1-\theta}{1-\mu}\right)^{1-\theta}
\right\}^{n}.
$$
\end{proof}
\\
A warning is in place, here.
A supermartingale was found to prove the validity of an upperbound of the probability of the event $\{\overline X\ge\theta\}$.
%
It is definitely not the case that $\{\sup_{h\ge0} M_k(h)\}_k$ is a test supermartingale.
\\
\\
Inspired by Poisson distributions we get
$$
g(t,\nu,\mu)
=
\frac{\nu^t\exp(-\nu)}{\mu^t\exp(-\mu)}=
\frac{\exp(ht)}{\exp(\mu(\exp(h)-1))}\hbox{ with }h=\log\left(\frac{\nu}{\mu}\right).
$$
One can show that $\E(g(X,\nu,\mu))\le1$ if $0\le X\le1$ a.s. 
such that $(\E(X)-\mu)(\nu-\mu)\le0$.
Here the corresponding test supermartingale can be expressed in terms of 
densities of Gamma distributions.
See Section \ref{Bounded case} for a testing method used in practice in statistical auditing which is said to be based on the Poisson distribution and which appears in the Audit Guide Audit Sampling \cite{AICPA}.


\section{Behavior of the test supermartingales $\{M_k(c)\}_k$}\label{Behavior}
Given $\mu>0$, we discuss the behavior of the test supermartingales $\{M_k(c)\}_k$
for null hypothesis $H_0:\E(X)\le\mu$ constructed according to 
(\ref{test martingale}) with $c_{k}=c$,
for all $k$,  $0\le c<1$. 
Let $X$ be a positive integrable random variable 
and $X_1,X_2,\ldots$ an iid sample of $X$.
We have
$$
\frac1k\log(M_k(c))=\frac1k\sum_{i=1}^k\log((1-c)+c\,X_i/\mu).
$$
Let $Z=X/\mu$, then $\E(Z)$ exists and 
$Z\ge0$ a.s.. 
Consider the function 
$$
\lambda(c)=\E(\log((1-c)+c\,Z)),\quad 0\le c<1.
$$
It is well defined since 
\begin{equation}\label{log bounds}
\log(1-c)\le\log((1-c)+c\,Z)\le -c+c\,Z\le Z.
\end{equation}
It also follows by dominated convergence that $\lambda(c)$ is continuous in $c$ for $0\le c<1$.
\begin{theorem}
$\lambda(c)$ is twice continuously differentiable and concave in $c$ for $0<c<1$.
$\lambda(0)=0$, $\lim_{c\downarrow0}\lambda'(c)=(\E(X)-\mu)/\mu$.
If $\P\{X=\mu\}<1$, it is strictly concave.
\end{theorem}
\begin{proof}
Notice that
\begin{equation}\label{derivatives}
\frac d{d c}\log((1-c)+cz)=\frac{z-1}{(1-c)+cz}
;\quad
\frac {d^2}{d c^2}\log((1-c)+cz)=-\frac{(z-1)^2}{((1-c)+cz)^2}.
\end{equation}
Moreover, for $0<c_0< c<c_1<1$,
\begin{equation}\label{log' bounds}
-\frac1{1-c_1}\le\frac{z-1}{(1-c)+cz}
\le \frac1{c_0}. %
\end{equation}
%
Thus differentiation of $\lambda(c)$ behaves decently
with respect to expected value in the range $c_0<c<c_1$ (see \cite[Thm. A.5.2]{D})
for all $0<c_0<c_1<1$, and we have with $Z=X/\mu$,
\begin{align*}
\lambda'(c)
&=
\E((Z-1)/(1-c+cZ))
\\
\lambda''(c)
&=
\E(-[(Z-1)/(1-c+c\, Z)]^2)\le0.
\end{align*}
By inequality (\ref{log' bounds}) and dominated convergence
it holds that $\lambda''(c)$ is continuous in $c$ for $0<c<1$.
Since $-1/(1-c_1)\le(z-1)/((1-c)+cz)\le z$ for $z\ge0$ and $0\le c\le c_1$
we obtain by dominated convergence
\begin{align*}
\lim_{c\downarrow0}\lambda'(c)
&=
\lim_{c\downarrow0}\E\left(\frac{Z-1}{1-c+c\, Z}\right)
=\E\left(\lim_{c\downarrow0}\frac{Z-1}{1-c+c\, Z}\right)=\E(Z-1)=\frac1\mu(\E(X)-\mu).
\end{align*}
If $\P\{X=\mu\}=\P\{Z=1\}<1$, then $\lambda''(c)<0$ for $0<c<1$
so that $\lambda(c)$ is strictly concave in $c$.
\end{proof}

\begin{remark}\label{L1}
From Equation (\ref{derivatives}) it follows that any realization 
of $L_n(c)=\log(M_n(c))$ (based on observations of $X_1,\ldots,X_n$) 
is concave in $c$.
Thus $M_n(c)$ is log-concave in $c$.
\end{remark}

\begin{corollary}\label{C2}
Suppose $\P\{X=\mu\}<1$.
If $\E(X)\le\mu$ and  $0<c<1$, then $\lambda(c)<0$, so that $\lim_{n\to\infty}M_n(c)=0$.
If $\E(X)>\mu$, then there is $c_{\max}>0$ such that for 
$0<c<c_{\max}$ we have $\lambda(c)>0$.
If $\lambda(c)>0$
then $\lim_{n\to\infty}M_n(c)=\infty$ a.s..
If moreover $\P\{X=0\}>0$, then $c_{\rm max}<1$.
\end{corollary}

\begin{proof}
Suppose $\E(X)\le\mu$ and $\P\{X=\mu\}<1$, then $\lambda'(0)\le0$ and 
from the strict concavity of $\lambda(c)$ in $c$, it follows that $\lambda(c)<0$ for all $0<c<1$.
According to the strong law of large numbers, it follows that
$\lim_{k\to\infty}\frac1k\log(M_k(c))=\lambda(c)<0$ a.s., 
and therefore that $\lim_{k\to\infty}M_k(c)=0$ a.s. 
(despite the fact that in case $\E(X)=\mu$ we have $\E(M_k)=1$ for all $k$, cf. \cite[Ex. 5.2.9]{D}).
If $\E(X)>\mu$, then $\lambda'(0)>0$ and there is $c_{\max}>0$ such that $\lambda(c)>0$
for $0<c<c_{\max}$.
If $\lambda'(0)>0$ we have  $\lim_{k\to\infty}\frac1k\log(M_k(c))=\lambda(c)>0$ a.s. and therefore
$\lim_{k\to\infty}M_k(c)=\infty$ a.s..
\\
If $\P\{X=0\}>0$ then $\lambda(c)\ge\log(1-c)\P\{X=0\}$ so that 
$\lim_{c\uparrow1}\lambda(c)=-\infty$. 
\end{proof}
\\
\\
The Corollary implies the following theorem.

\begin{theorem}[Consistency]\label{Consistency}
Suppose $\E(X)>\mu$ and $X\ge0$ a.s..
Then there is $c_{\max}>0$ such that $\lim_{n\to\infty} M_n(c)=\infty$ a.s. for $0<c<c_{\max}$.
Let $\pi$ be a probability density on $[0,1]$ such that for all $0<b<1$ there is $0<a<b$
such that $\int_a^b\pi(c)dc>0$.
Then the integrated test supermartingale $\{M_n(\pi)\}_n$
satisfies $\lim_{n\to\infty} M_n(\pi)=\infty$ a.s.,
and the test based on $\{M_n(\pi)\}_n$ is consistent, i.e. the power of the test is 1.
\end{theorem}

\begin{proof}
The first claim follows from Corollary \ref{C2}.
Let $0<a<b<c_{\max}$, such that $p=\int_a^b\pi(c)dc>0$.
We have $\lim_{n\to\infty} M_n(c)=\infty$ a.s. for $c=a,b$.
Be given any $R>0$,
let $N$ be such that $M_N(a)>R/p$ and $M_N(b)>R/p$. 
Since $M_N(c)$
is a log-concave function in $c$ we have $M_N(c)>R/p$ for all $c\in[a,b]$, so that
$M_N(\pi)>(R/p)\cdot p=R$.
\end{proof}

\subsubsection*{Performance depending on $X$}
In the remainder of the section we will give some tools to evaluate the performance of the test
using test supermartingales as above using an iid sample of $X$, 
in case the rv $X$ is known.
We are interested in the mean value of the sample size $N$ needed to reject the null hypothesis.
\begin{lemma}\label{moments}
Let $Y$ be an integrable random variable such that $Y\ge c$ a.s. for some $c>0$.
Then the moment generating function 
$\phi(h)=\E(\exp(h\,\log(Y)))=\E(Y^h)$ is finite for all $h\le1$.
In particular all moments of $\log(Y)$ are finite.
\end{lemma}

\begin{proof}
For $h=1$ we have $\exp(h\log(Y))=Y$ which is integrable, so that $\phi(1)<\infty$.
For $h\le0$ we have $0\le\exp(h\log(Y))\le \exp(h\log(c))=c^h$ so that $\exp(h\log(Y))$ is integrable,
and $\phi(h)<\infty$ for $h\le0$.
Together this implies that $\phi(h)<\infty$ for all $h\le1$, 
that $\phi(h)$ is infinitely differentiable in the region $h<1$, and that the $k$-th order moment of $\log(Y)$ equals
the $k$-th derivative of $\phi(h)$ at $h=0$ (see e.g. \cite[Sec. 21]{B}).
\end{proof}
\\
\\
Let $X$ be a rv with $X\ge0$ a.s. such that $\E(X)=\nu>\mu$.
Let
$0<c<1$ such that with 
$Y=(1-c)+c\,X/\mu$ 
we have $\omega=\E(\log(Y))>0$. 
Let $\rho^2=\Var(\log(Y))$ which is finite because of Lemma \ref{moments}.
Given an iid sample $X_1,X_2,\ldots$ of $X$, let $M_n(c)$ be the test supermartingale
with $i$-th factor $Y_i=(1-c)+c\,X_i/\mu$.
For $0<\alpha<1$
let 
$$
N=N_\alpha=\inf\{n\mid M_n(c)\ge1/\alpha\}
=\inf\{n\mid \log(M_n(c))=\sum_{i=1}^n\log(Y_i)\ge\log(1/\alpha)\}.
$$
By Corollary \ref{C2} we have $N_\alpha<\infty$ a.s.. 
From Theorem 2.5 in \cite[Ch. 2.5]{CRS} on first passage times,
it follows that for all $0<\alpha<1$ we have $\E(N_\alpha)<\infty$
and $\E(N_\alpha^2)<\infty$ and 
\begin{equation}\label{renewal theorem}
\lim_{\alpha\downarrow0}\frac{\E(N_\alpha)}{\log(1/\alpha)}=\frac1\omega
\hbox{ and }
\lim_{\alpha\downarrow0}\frac{\Var(N_\alpha)}{\log(1/\alpha)}=\frac{\rho^2}{\omega^3}.
\end{equation}
Lindeberg's condition (2.29) \emph{l.c.} holds as $\{\log(Y_i)\}_i$ is an iid sample of square integrable rv $\{\log(Y)\}_i$.
Since $\E(N_\alpha)<\infty$ 
Wald's equation (\cite[Thm. 4.1.5]{D}) gives the more precise expression
\begin{equation}\label{Wald's first equation}
\E(\log(M_{N_\alpha}(c)))=\E(\log(Y))\E(N_\alpha)=\omega\,\E(N_\alpha).
\end{equation}
%
Recall Lorden's inequality (\cite{L}
) for the expected excess,
which is independent of $\alpha$:
$$
\E(\log(M_{N_\alpha}(c)))-\log(1/\alpha)
\le 
\frac{\E((Z^+)^2)}{\E(Z)}
\le
\omega+\frac{\rho^2}{\omega}.$$
With Wald's equation (\ref{Wald's first equation}) this gives
$\log(1/\alpha)/\omega\le\E(N_\alpha)\le\log(1/\alpha)/\omega+1+\rho^2/\omega^2.$

\subsubsection*{Dependence on the mean and variance of $X$}

In this subsection we will consider in more detail the performance
of the test supermartingales $\{M_n(c)\}$ for $H_0:\E(X)\le\mu$, based on a sample
of $X$ with $\E(X)=\theta>\mu$ and finite variance $\Var(X)\le\sigma^2$.
Let $X_0$ be a 2-point distribution with values 0 and $\tau=(\theta^2+\sigma^2)/\theta$
and $\P\{X=\tau\}=\theta/\tau$ so that $\E(X)=\theta$, $\E(X^2)=\tau^2\theta/\tau=\theta^2+\sigma^2$
and $\Var(X_0)=\sigma^2$.
We will show that among rv's $X$ as above the test supermartingales perform worst for $X_0$.

\begin{lemma}\label{Jensen++}
Let $\ell:[0,\infty)\to \R$ be a three times continuous differentiable function
such that $\ell''(x)\le0$ and 
$\ell'''(x)\ge0$ for all $x\ge0$.
Then $\E(\ell(X))\ge \E(\ell(X_0))$.
\end{lemma}

\begin{proof}
Let $h(x)=u+vx+wx^2$ be such that 
$h(0)=\ell(0)$, $h(\tau)=\ell(\tau)$ and $h'(\tau)=\ell'(\tau)$ then $\ell(X_0)=h(X_0)$ and
we will show $\ell(x)\ge h(x)$ for all $x\ge0$.
Function $L(x)=\ell(x)-h(x)$ satisfies the ordinary differential equation
$L'''(x)=\ell'''(x)$ for $x\ge0$ and the 'boundary' conditions $L(0)=L(\tau)=L'(\tau)=0$.
The corresponding Green's function is
$$
G(x,u)=
\begin{cases}
\frac12(x-u)^2, \hbox{ if } \tau<u<x
\\
\frac12\frac{u^2}{\tau^2}(x-\tau)^2, \hbox{ if } 0<u<\tau \hbox{ and } u<x
\\
(\tau-u)\frac x\tau((u-x)+\frac12\frac x\tau(\tau-u)), \hbox{ if } 0<x<u<\tau
\\
0, \hbox{ if } \tau<u \hbox{ and } x<u 
\end{cases}
$$
and $L(x)=\int_0^{\max(x,\tau)}G(x,u)\ell'''(u)du$.
Since $G(x,u)\ge0$ for all $x>0$ and $u>0$, and $\ell'''(x)\ge0$ it follows that $L(x)\ge0$.
In particular $L''(\tau)=\ell''(\tau)-h''(\tau)=\ell''(\tau)-2w\ge0$ so that
$w\le\frac12\ell''(\tau)\le0$.
It follows that $\E(\ell(X))\ge\E(h(X))=u+v\theta+w(\theta^2+\Var(X))
\ge u+v\theta+w(\theta^2+\sigma^2)=\E(h(X_0))=\E(\ell(X_0))$.
\end{proof}

\begin{theorem}\label{ThmDKL}
Let $\theta>\mu$ and $\sigma^2>0$.
With $\tau=(\theta^2+\sigma^2)/\theta$  
the random variable $X_0$ with values $0$ and $\tau$ such that 
$\P\{X_0=\tau\}=\theta/\tau$ and 
$\P\{X_0=0\}=1-\theta/\tau$ satisfies $\E(X_0)=\theta$
and $\Var(X)=\sigma^2$.
Furthermore, let
$$\lambda_0(c)=\E(\log((1-c)+cX_0/\mu))
=
(1-\frac\theta\tau)\log(1-c)
+\frac\theta\tau\log(1-c+c\tau/\mu).$$
Its maximum is $c_0=(\theta-\mu)/(\tau-\mu)$ and the maximum value is
\begin{equation}\label{DKL}
\left(1-\frac\theta\tau\right)\log\left(\frac{\tau-\theta}{\tau-\mu}\right)
+
\frac\theta\tau\log\left(\frac{\theta}{\mu}\right)=D_\KL(\Alt(\theta/\tau)\|\Alt(\mu/\tau)),
\end{equation}
where $D_\KL(P\|Q)$ denotes the Kullback-Leibler divergence from $Q$ to $P$.
\\
Let $X$ be a positive random variable with $\E(X)=\theta$ and variance $\Var(X)\le\sigma^2$.
Then $\lambda(c)=\E(\log((1-c)+cX/\mu))\ge\lambda_0(c)$ and
the maximum $c_{{\rm opt}}$ of $\lambda(c)$ is at least 
$c_0$.
\end{theorem}
\begin{remark}\label{bddcase}
If the rv $X$ with $\E(X)=\nu\ge\theta>\mu$ is bounded, say $0\le X\le\gamma$ a.s.,
then $\E(X^2)\le\E(X\gamma)=\nu\gamma$ so that 
$\Var(X)\le\nu(\gamma-\nu)$ and
$\tau=(\nu^2+\nu(\gamma-\nu))/\nu=\gamma$. 
We find the lower bound 
$(\nu-\mu)/(\gamma-\mu)\ge(\theta-\mu)/(\gamma-\mu)$ for $c_{{\rm opt}}$ independent of $\Var(X)$.
\end{remark}

\begin{proof}[Proof of Theorem \ref{ThmDKL}]
Only the claims about $\lambda(c)$ need
 explanation.
The inequality $\lambda(c)\ge\lambda_0(c)$
is based on Lemma \ref{Jensen++} applied to $\ell(x)=\log(1-c+cx/\mu)$. 
Lemma \ref{Jensen++} applied to 
$\frac{\partial}{\partial c}\ell(x)=(-1+x/\mu)/(1-c+cx/\mu)$
leads to 
$$\left.\frac{\partial}{\partial c}\right|_{c=c_0}\E(\log(1-c+cX/\mu))
\ge
\left.\frac{\partial}{\partial c}\right|_{c=c_0}\E(\log((1-c)+cX_0/\mu))=0.$$
To conclude, since $\lambda(c)$ is a convex function in $c$. 
and its derivative at $c_0$
is non-negative, its maximum $c_{{\rm opt}}$ satisfies $c_{{\rm opt}}\ge c_0$.
\end{proof}
\\
\\
We would like to stress the fact that the performance of the above tests 
depends on the variance of the random variable $X$, in contrast to the tests based on test supermartingales using factors as proposed in Section \ref{SPRT}. 
Given that $\E(X)=\nu>\mu$, by Jensen's inequality
the quantity $\E(\log((1-c)+cX/\mu))$ 
is maximal for the constant rv $X=\nu$, so that in that case $\E(N_\alpha)$ will be minimal. 
On the other hand, consider as in Section \ref{SPRT}  
rv's $X$, such that the moment generating
function exists and satisfies the inequality 
$$\E(\exp(hX))\le\varphi(h,\E(X))$$
for $h$ in some open interval containing 0. 
\\
Suppose for $h\ge0$ that $\varphi(h,\nu)$ is increasing in $\nu$. 
Then 
we may build a test supermartingale for $H_0:\E(X)\le\mu$ based on factors
$g(X_i,h,\mu)=\exp(hX_i)\varphi(h,\mu)^{-1}$, where $X_1,X_2,\ldots$ is an iid sample of $X$.
Define 
$$\lambda(h)=\E(\log(g(X,h,\mu)))=\E(hX-\log(\varphi(h,\mu)))=h\E(X)-\log(\varphi(h,\mu)).$$
Assume $\varphi(0,\mu)=1$ 
and $\left.\frac\partial{\partial h}\right|_{h=0}\varphi(h,\mu)=\mu$.
Suppose $\E(X)>\mu$ then there exists $h>0$ such that $\omega=\lambda(h)>0$.
It follows from equations (\ref{renewal theorem}) that 
$\lim_{\alpha\downarrow0}\E(N_\alpha)/\log(1/\alpha)$
only depends on $\E(X)$ through 
$\omega=\lambda(h)=h\E(X)-\log(\varphi(h,\mu))$.
\\
\\
We will work this out for the example $\varphi(h,\mu)=(1-\mu)+\mu\,\exp(h)$.
Recall that
random variables $X$ with values in $[0,1]$ satisfy condition $\E(\exp(hX))\le\varphi(h,\E(X))$.
\\
If $\E(X)=\theta>\mu$, we find maximum $h=\log(\theta(1-\mu)/(\mu(1-\theta)))$ for $\E(\log(g(X,h,\mu)))$
with maximum value $\omega_m=\theta\log(\theta/\mu)+(1-\theta)\log((1-\theta)/(1-\mu))$
which only depends on $\E(X)=\theta$.
In case $X$ is Bernoulli distributed, the resulting test martingale coincides with
test martingale (\ref{test martingale}) with test parameters 
$c_k=c=(\theta-\mu)/(1-\mu)$.
Given $\E(X)=\theta$ and $0\le X\le1$ a.s., so that $\Var(X)\le\theta(1-\theta)$,
according to Theorem \ref{ThmDKL} the latter test martingale performs worst if $X$ is Bernoulli  $\Alt(\theta)$ distributed.
\\
For the constant variable $X=\theta$, with the above test parameter $c=(\theta-\mu)/(1-\mu)$ we would get 
$\omega_\theta=\log((1-\theta)/(1-\mu)+((\theta-\mu)/(1-\mu))\theta/\mu)$.
For this $X$ even parameter $c=1$ is allowed giving  
$\omega_1=\log(\theta/\mu)$.
For $\mu=1-0.05$ and $\theta=1-0.02$ we have $\omega_m=0.01214$ to be compared with 
$\omega_\theta=0.01877$
resp. $\omega_1=0.03109$, yielding an improvement in the expected sample size by roughly 35\%, resp. 60\%.

\section{Application in audit sampling}\label{Bounded case}
In this section we consider the test supermartingales 
(\ref{test martingale}) and (\ref{its}) developed in Section \ref{H0},
under the assumption that the random variable $X$
is bounded, that is, there are $\tau_0<\tau_1$ such that $\tau_0\le X\le\tau_1$ a.s..
By transforming $X$ to $T=(X-\tau_0)/(\tau_1-\tau_0)$ or $T=(\tau_1-X)/(\tau_1-\tau_0)$
we will restrict our attention to
 null hypotheses of the form
$H_0:\E(T)\ge\mu$ under the condition that $0\le T\le1$ a.s..
This is the context closest to the intended application in audit sampling.
\\
\\
We summarize some relevant facts of the test procedure.
Suppose that $0\le T\le1$ a.s. and $0\le \E(T)=\nu<\mu<1$.
We consider test supermartingales $\{M_n(c)\}_n$ for $H_0:\E(1-T)\le1-\mu$,
or equivalently for $H_0:\E(T)\ge\mu$.
Given a random sample $T_1,T_2,\ldots$ of $T$, test parameter $c$ such that $0\le c<1$ 
and time $n$, they are defined by
$$M_n(c)=\prod_{i=1}^n(1-c+c(1-T_i)/(1-\mu)).$$
The test starts  with the choice of $c$ at time 0 and and the sequential calculation of $M_n(c)$ and leads to rejection of $H_0$ 
if there is $n=N_\alpha$ for which $M_n(c)\ge1/\alpha$ at which time one may stop sampling.
\\
It follows from Theorem \ref{ThmDKL} that worst case in terms of the expectation of $N_\alpha$, is achieved for Bernoulli variable
$T^\nu$ with success probability $\nu$, $T^\nu\sim\Alt(\nu)$.
For this variable the optimal choice of $c$ is $c_0=(\nu-\mu)/\mu$.
Moreover, for the original variable $T$ the optimal value of $c$ satisfies $c\ge c_0$.
If $c$ satisfies $\E(\log(1-c+c(1-T)/(1-\mu)))
>0$, 
the power of the test is 1.
This is the case at least for $0<c\le c_0$.
If  $\E(\log(1-c+c(1-T)/(1-\mu))\ge\omega>0$ for some known $\omega$,
the first item in the theorem on first passage times given in Formula (\ref{renewal theorem}) leads to approximate upper bound $\log(1/\alpha)/\omega$ for the mean time $\E(N_\alpha)$ at which $H_0$ can be rejected when using
parameter $c$.
This holds for $c=c_0$ with $\omega=\E(\log(1-c_0+c_0(1-T^\nu)/(1-\mu)))=D_\KL(\Alt(\nu)\|\Alt(\mu))>0$.
\\
\\
The condition at rejection time $n$ that $M_n(c)\ge1/\alpha$ ensures that the
significance level of the testing procedure is $\alpha$.
In the next theorem we put some bounds on the actual probability of Type I error.
\begin{theorem}\label{size}
Let $T$ be a random variable such that $0\le T\le1$ a.s., $\E(T)=\mu$
and $\P(T=\mu)<1$. 
Let $0<c<1$
and consider the test supermartingale $\{M_k(c)\}_k$ for $H_0:\E(T)\ge\mu$ with 
multiplication factor $((1-c)+c\,(1-T_k)/(1-\mu))$
at time $k$.
The probability of Type I error of the test is at most $\alpha$ but greater than 
$\alpha/(1-c+c/(1-\mu))>\alpha(1-\mu)$.
\end{theorem}
In particular, if $\E(T)=\mu$ and $\mu$ is small,
the null hypothesis $H_0:\E(T)\ge\mu$ wil be rejected with
probability close to  (but not more than) $\alpha$.
On the other hand, if $\mu$ is close to 1, the probability of Type I error may be
considerably smaller than $\alpha$.

\begin{proof}
When applied to a random sample of $T$, the process $\{M_k\}_k=\{M_k(c)\}_k$ is a martingale.
Let $N=\inf\{n\mid\log(M_n\ge\log(1/\alpha)\}$, so that $N=\infty$ if $M_n<1/\alpha$ for all $n$.
Then $\{M^\alpha_k\}_k=\{M_{k\wedge N}\}_k$ is a stopped martingale, and therefore a martingale.
Let $M^\alpha_\infty=\lim_{k\to\infty}M^\alpha_k$, 
then (see Cor. \ref{C2})  $M^\alpha_\infty$ takes values in the 
set $\{0\}\cup[\alpha^{-1},\alpha^{-1}(1-c+c/(1-\mu)))$. 
Because of dominated convergence we have
$1=\E(M_0)=\lim_{k\to\infty}\E(M^\alpha_k)=\E(M^\alpha_\infty)$
and 
$\alpha^{-1}\P\{M^\alpha_\infty\ge\alpha^{-1}\}\le\E(M^\alpha_\infty)<
\alpha^{-1}(1-c+c/(1-\mu))\P\{M^\alpha_\infty\ge\alpha^{-1}\}$.
\end{proof}

\subsubsection*{Average sample number}
We consider an example of the performance of tests as discussed in this section.
We took $\mu=0.05$, significance level $\alpha=0.05$ 
and considered the necessary sample number for rejection of $H_0:\E(T)\ge\mu$ for different 
$T$-distributions with $0\le T\le1$, $\E(T)=0.02$ and 
test supermartingales $\{M_k(c)\}_k$ with $0\le c\le 1$.
The results are compiled in Table \ref{T1}.
The last line starting with $[0.6,1]$ is based on the integrated test martingale 
$\{M_k(\pi)\}_k$ 
using as test measure the uniform probability density $\pi=2.5$ on the interval $[0.6,1]$, 
possibly based on a strong conviction
that $\E(T)\le\nu=0.02$, implying that the optimal $c$ is not less than $c_0=(\mu-\nu)/\mu=0.6$.
The average sample number and standard deviations in each instance are based on 1000 test runs.
The results for fixed $c$ are in close agreement with the approximations
of $\E(N_\alpha)$ and $\Var(N_\alpha)$ following from (\ref{renewal theorem}).
This is mainly due to the small excess
of $\log(M_n(c))$ over $\log(1/\alpha)$ at decision time $n=N_\alpha$ of at most $\log(1/(1-\mu))\approx0.05$
with respect to $\log(1/\alpha)\approx3$ (see Theorem \ref{size}).
\\
\noindent
\begin{table}[h]
\begin{tabular}{lcc}
$T:$&\multicolumn{2}{c}{$\Alt(0.02)$}\\
\phantom{0}$c$&\multicolumn{1}{c}{mean}&sd\\
0.2 & 516.2 & 127.4 \\
0.4 & 294.1 & 124.4 \\
0.6 & 245.9 & 169.2 \\
0.8 & 357.7 & 510.5 \\
1 & $\infty$ & -- \\
$[0.6,1]$ & 287.6&253.2
\end{tabular}
\begin{tabular}{cc}
\multicolumn{2}{c}{$5T\sim\Alt(0.10)$}\\
\multicolumn{1}{c}{mean}&sd\phantom{0}\\
482.7 & 45.0 \\
245.5 & 32.7 \\
166.7 & 27.7 \\
127.4 & 24.9 \\
104.0 & 23.4 \\
124.6 & 25.4
\end{tabular}
\begin{tabular}{rr}
\multicolumn{2}{c}{$\hbox{Beta}(0.02,0.98)$}\\
\multicolumn{1}{c}{mean}&sd\phantom{0}\\
495.3 & 81.7 \\
261.3 & 67.8 \\
186.9 & 68.0 \\
156.1 & 79.8 \\
166.2 & 155.0 \\
162.4&85.5
\end{tabular}
\begin{tabular}{rr}
\multicolumn{2}{c}{$\hbox{Beta}(2,98)$}\\
\multicolumn{1}{c}{mean}&sd\phantom{0}\\
476.6 & 10.1 \\
239.6 & 7.1 \\
160.5 & 5.9 \\
121.0 & 5.1 \\
97.2 & 4.5 \\
117.7&5.3
\end{tabular}
\begin{tabular}{r}
$0.02$\\
\multicolumn{1}{c}{=}\\
476 \\
239 \\
160 \\
121 \\
97 \\
117
\end{tabular}
\caption{Average sample numbers (mean) and their standard deviations (sd), 
in each case based on 1000 tests of $H_0:\E(T)\le0.05$ with significance level $\alpha=0.05$.}
\label{T1}
\end{table}
\\
We see confirmed that the optimal $c$ depending on the distribution of $T$ is some number greater than
$(\mu-\nu)/\mu=0.6$ (see Theorem \ref{ThmDKL}).
Recall that it is required that $c\le1$ 
in order to have a test supermartingale.
Notice that, except for the Bernoulli distribution, the integrated test supermartingale
integrated over the interval $[(\mu-\nu)/\mu,1]$ outperforms
the test supermartingale corresponding to the fixed test parameter $c=(\mu-\nu)/\mu$.
If one has no idea about $\E(T)$ other than $\E(T)<\mu$, then the integrated test supermartingale, 
integrated uniformly over $[0,1]$,
is a suitable choice.
\\
One may compare the results in Table \ref{T1} with a common practice in financial auditing
as specified in 
Appendix C-1 of the Audit Guide Audit Sampling \cite{AICPA}. 
In it, the $T$-values are referred to as \emph{taints}, and $\E(T)$ corresponds to the \emph{misstatement as a fraction of the population}.
This Appendix contains Table C-1: \emph{Monetary Unit Sample Size Determination Tables}, said to be  \emph{based on the Poisson distribution}.
With \emph{risk of incorrect acceptance} $\alpha$, \emph{tolerabele misstatement as a fraction}, $\mu$, and 
\emph{expected misstatement as a fraction}, $\nu$, as above, Table C-1
gives the optimal integer solution $n$ of the inequality $n\nu\ge Q(1-\alpha,1+n\mu)$
where $Q(p,a)$ denotes the $p$-quantile of the Gamma distribution with shape parameter $a$ and unit scale parameter.
For $\mu,\alpha,\nu$ as chosen in the simulations we find  $n=162$.
If the total sum of $T$-values in the sample does not exceed $n\nu=3.24$ the auditor
may conclude that the population is not misstated by a fraction more than $\mu$.
Notice that, whatever properties this testing procedure has, if in fact $\E(T)=\nu$, the critical level $n\nu$ is the expected value of the total sum of $T$-values in the sample 
and therefore,  for small $\alpha$ leading to large $n$,
the auditor will be successful 
with probability close to 0.50. 
If unfortunately, the total sum of $T$-values in the sample exceeds 3.24, the evidence that the misstatement as a fraction is more than $\mu$ is still rather weak and
further evidence shall be collected in \emph{providing support for the conclusions on which to base one's opinion}. 
Such open ends require a thorough preparation of an audit
if one wants to assign definite properties to the audit procedure.
~\\
\\
Testing $H_0:\E(T)\ge\mu$ in the conviction that $\E(T)\le\theta$
reminds of the context of acceptance sampling, where power 1 tests are not customary.
We will end this section by elaborating the following idea.

\begin{remark}\label{inverse tm}
Suppose $0\le T\le 1$ a.s. and $0<c<1$. 
Let $\{M_k(c)\}_k$ be test supermartingale for $H_0:\E(X)\ge\mu$ as constructed before
with factors $((1-c)+c\,(1-T_k)/(1-\mu))$.
Then $\{(M_k(c))^{-1}\}_{k=0}^\infty$ is a test supermartingale for null hypothesis 
$H_1:\E(T)\le\theta$ for 
$\theta=(1-c)\mu$ (or $c=(\mu-\theta)/\mu$).
\\
If $\E(T)\le\theta$ we have 
$\lim_{k\to\infty}M_k(c)=\infty$ a.s..
\\
If $\E(T)\ge\mu$ and $\P\{T=\mu\}<1$ then we have 
$\lim_{k\to\infty}(M_k(c))^{-1}=\infty$ a.s..
\end{remark}
\begin{proof}[Proof of Remark \ref{inverse tm}]
Suppose $\E(T)=\nu$.
Let 
$g(t)=(1-c+c\,(1-t)/(1-\mu))^{-1}$. 
In particular $g(t)$ is convex in $t$,
so that $g(t)\le (1-t)g(0)+tg(1)$ and 
the maximal value of $\E(g(T))$ is attained at 
the  Bernoulli rv $T_0$
with mean value $\nu$.
One easily checks that $\E(g(T_0))\le1$ if $\nu\le(1-c)\mu=\theta$.
It follows that $\{M_k(c)^{-1}\}_k$ is a test supermartingale for $H_0:\E(T)\le\theta$.
\\
If $\E(T)\ge\mu$ and $\P\{T=\mu\}<1$ then 
it follows from Corollary \ref{C2} applied to $X=1-T$ that 
$\lambda(c)=\E(\log((1-c)+c(1-T)/(1-\mu)))<0$ so that $\lim_{k\to\infty}M_k(c)=0$ a.s..
On the other hand, if $\E(T)\le\theta$, by Remark \ref{bddcase}
we have $\lambda(c)=\E(\log((1-c)+c(1-T)/(1-\mu))>0$ so that by Corollary \ref{C2} we have 
$\lim_{k\to\infty}M_k(c)=\infty$.
\\
Notice that if $\theta<\E(T)<\mu$ it is possible that $\lambda(c)=0$ and then the
expected decision time is infinite.
\end{proof}
\\
\\
Consider 
a variable $T$ such that
$0\le T\le1$ a.s..  
Consider some tolerance level $\mu$ (`Lot Tolerance Percent Defective') which should not be exceeded, 
where the supplier is able to provide quality level $\theta<\mu$ (`Acceptable Quality Level') at which he wishes
that the lot will be accepted.
The inspector and the supplier agree on the following test procedure specification.
If the hypothesis $H_0:\E(T)\ge\mu$ can be rejected at significance level $\alpha$,
the inspector will accept the lot.
If the hypothesis $H_1:\E(T)\le\theta$ can be rejected at significance level $\beta$,
the inspector will reject the lot.
Rejection should go together with some provision to protect the
average quality level of accepted lots (`Average Outgoing Quality').  
Consider the test supermartingale for $H_0:\E(T)\ge\mu$
$$
M_n=\prod_{i=1}^n((1-\frac{\mu-\theta}\mu)+\frac{\mu-\theta}\mu\times\frac{1-T_i}{1-\mu})
$$
Let decision time $n$ be the first time at which $M_n\ge A$ or $M_n\le B$.
Then reject $H_0$ if $M_n\ge A$ and reject $H_1$ if $M_n\le B$.
According to Remark \ref{inverse tm}, 
$M_n^{-1}$ is a test supermartingale for $H_1$.
For the procedure we can safely take $A=\alpha^{-1}$ and $B=\beta$,
but there is an opportunity for improvement (cf. \cite{W}).
\\
Let $\alpha^*$ (depending on $T$) be the probability to reject $H_0$ if $\E(T)=\mu$.
Let us ignore the overshoot over $A$ at decision time $n$.
This is justified if 
$(\mu-\theta)/(1-\mu)$ is small.
Since $\{M_n\}_n$ is martingale, we have the (approximate) equality
$1\approx\alpha^* A+(1-\alpha^*)B$, so that $\alpha^*\approx(1-B)/(A-B)$.
If $\E(T)\le\theta$, then $\{M_n^{-1}\}_n$ is a supermartingale so that
we get 
$1\ge(1-\beta^*)A^{-1}+\beta^* B^{-1}$, that is $\beta^*\le(1-A^{-1})/(B^{-1}-A^{-1})$.
In particular, still ignoring overshoots at decision time, $A=(1-\beta)/\alpha$ and $B=\beta/(1-\alpha)$
satisfy the specification.

\section{Confidence regions}\label{conf reg}

As one will have noticed we did not include a provision in our tests to avoid infinite sample size,
as could easily happen e.g.~in case $H_0$ is satisfied.
In practice it may be a more important issue to find a suitable confidence lower bound or
a confidence interval.
We will investigate ways to reuse the sample without loss of confidence 
if at any time one decides to switch from testing 
$H_0:\E(T)\le\mu_0$ to determining a confidence lower bound.
One may consult Waudby-Smith and Ramdas, \cite{W-SR}, for a thorough treatment of the construction of confidence
intervals using test supermartingales, especially from the perspective of large sample sizes.

\subsubsection*{Confidence lower bounds}
Choose a confidence level $(1-\alpha)$ with $0<\alpha<1$, for example $\alpha=0.05$.
We will construct an adapted process $\{B^l_k\}_k$ of $(1-\alpha)$-confidence lower bounds
such that even $\P\{\forall k:\E(T)> B^l_k\}\ge1-\alpha$.
\\
\\
Suppose for each $\mu>0$ we have maintained a process $\{M^\mu_k\}_k$, such that
$\{M^\mu_k\}_k$ is test supermartingale for the hypothesis $H_0:\E(T)\le \mu$ for a given integrable rv $T$ such that $T\ge0$ a.s..
Suppose moreover that any realization of $\{M^\mu_k\}_\mu$ 
is continuous and decreasing in $\mu$
for all $k$. 
%
We will call such a family 
$\{M^\mu_k\}_{k,\mu}$ a \emph{decreasing} family of test supermartingales for
the family of null hypotheses $\{H_0:\E(T)\le\mu\}_\mu$.
When used as a tool in the determination of confidence regions an additional useful property
is that $M^\mu_k$ is convex in $\mu$.
We will refer to such a family as a \emph{convex decreasing} family of test supermartingales.
We give some examples:
\begin{remark}\label{suitable families}
Let $T_1,T_2,\ldots$ be an iid sample of $T$.
Let $\{M^\mu_k\}_{k,\mu}$ be the family of test supermartingales 
where 
\begin{equation}\label{Mkmu}
M_k^\mu=\prod_{i=1}^k((1-c_i)+c_i\,T_i/\mu)
\end{equation} 
as in (\ref{test martingale}),
where for each $i$, $c_i$ depends on $T_1,\ldots,T_i$ but not on $\mu$.
Then the family $\{M^\mu_k\}_{k,\mu}$ of test supermartingales is a convex decreasing family.
\\
In particular, if $\pi$ is a probability density on $[0,1]$
(determined independently of the sample) 
the family of integrated test supermartingales
$\{M_k^\mu(\pi)\}_k$ is a convex decreasing family, where 
\begin{equation}\label{Mkmupi}
M_k^\mu(\pi)=\int_0^1M_k^\mu(c)\pi(c)dc 
\hbox{ with }M_k^\mu(c)=\prod_{i=1}^k((1-c)+c\,T_i/\mu)
\end{equation}
Families of this form will be called \emph{integrated decreasing family.}
\\
If $0\le T\le1$ a.s., the following (non convex) decreasing families of test supermartingales 
for $0<\Delta<1$ may also be useful in finding a confidence lower bound
greater than $\E(T)-\Delta$, cf. Theorem \ref{ThmDKL}.
\begin{equation}
M_k^\mu(\Delta)=\prod_{i=1}^k\left((1-\frac{\Delta}{1-\mu})+\frac{\Delta}{1-\mu}\,\frac{T_i}\mu\right)
,\hbox{ for }0\le\mu\le1-\Delta.
\end{equation}
\end{remark}
Suppose $\{M^\mu_k\}_{k,\mu}$ is a decreasing family.
Let $B^l_k$ be the statistic 
$$B^l_k=-\infty \hbox{ if } M^\mu_k<1/\alpha\hbox{ for all }\mu>0,\hbox{ and }
B^l_k=\sup\{\mu\mid M^\mu_k\ge1/\alpha\}\hbox{ otherwise. }
$$
If $B^l_k=\mu^-\ge0$, then it follows from continuity that $M_k^{\mu^-}\ge1/\alpha$.
%
%
Notice that for $\mu\ge0$ we have the equivalence $\mu\le B_k^l \Leftrightarrow M^\mu_k\ge1/\alpha$.
In particular, if $\E(T)=\mu$,
then\\  $\P\{\exists k:\mu\le B^l_k\}=\P\{\exists{k}:M^\mu_k\ge1/\alpha\}\le\alpha$,
so that $\P\{\forall k:\mu> B^l_k\}\ge1-\alpha$.

\begin{theorem}
Suppose $\{M^\mu_k\}_{k,\mu}$ is a decreasing family of test supermartingales for the null hypotheses
$H_0:\E(T)\le\mu$, where $T\ge0$ a.s..
\\
Then the statistics $B^l_k$ satisfy
$\P\{\forall k:\E(T)>B^l_k\}\ge1-\alpha$.
\end{theorem}
Thus, at time $k$, $\max\{B^l_j\mid j\le k\}$ is a $(1-\alpha)$-confidence lower bound.
Moreover, if its value at time $k$ is not convenient, one may continue sampling
in the hope to find
a better lower bound without losing confidence.
\\
Also, if initially the investigation of $T$ started off by testing $H_0:\E(T)\le\mu_0$ for 
a fixed $\mu_0$ at significance level $\alpha$ using some test supermartingale $\{M_k\}_k$ 
(e.g. as in (\ref{test martingale}) or (\ref{its}))
then at any time $k$ 
one may change one's mind 
and \emph{reuse the sample} to construct a $(1-\alpha)$-confidence lower bound by means of a decreasing family $\{M_k^\mu\}_{k,\mu}$ provided $M_i^{\mu_0}=M_i$ for $i\le k$
(equality as random variables).
The reason is that under these conditions rejection of $H_0$ at significance level $\alpha$ 
is equivalent to finding a $(1-\alpha)$-confidence lower bound that is at most
$\mu_0$.
%
%
\\
\\
\noindent
We present the following example using $\{M_k^\mu(1)\}_{k,\mu}$ 
as defined in (\ref{Mkmupi}):
\begin{remark}
Suppose $T\sim\Alt(\nu)$, let $T_1,T_2,\ldots$ be an iid sample of $T$
and define $S_k=T_1+\cdots+T_k$.
\\
Consider the test supermartingales $\{M_k^\mu(1)\}_{k,\mu}$
constructed with the factors $f(T_i)=T_i/\mu$.
Then $M_k^\mu=(1/\mu)^{k}$ if $S_k=k$, $M_k^\mu=0$ if $S_k<k$.
The corresponding $(1-\alpha)$-confidence lower bound is 
$B^l_k=\alpha^{1/k}$ if $S_k=k$ and otherwise $M_k^\mu=0$ for all $\mu$ 
so that
$B^l_k=-\infty$.
Let $K$ be the stopping time defined by $K= k$ if $T_i=1$ for $i<k$ and $T_k=0$.
Then $\mu^l=\max_k B^l_k=B^l_{K-1}=\alpha^{1/(K-1)}$ 
(or $-\infty$ if $K=1$) is a $(1-\alpha)$-confidence lower bound.
\end{remark}
%
%
%
%
Consider a convex decreasing family $\{M^\mu_k\}_{k,\mu}$ as defined by equation 
(\ref{Mkmu}) of Remark \ref{suitable families}.
Suppose for given $k$ one has observed $T_1=t_1,\ldots,T_k=t_k$.
Since $\log$ is a concave function, Jensen's inequality implies that for
$m=\frac1k\sum_{i=1}^k((1-c_i)+c_i\,t_i/\mu)$ we have $\frac1k\log(M^\mu_k)\le \log(m)$.
If there is $i\le k$ such that $c_it_i>0$ we may solve equality $m=1$ for $\mu$ and get
$$
M^\mu_k \le 1 \hbox{ if }\mu=\frac{\sum_{i=1}^k c_i t_i}{\sum_{i=1}^k c_i}.
$$ 
It follows that for integrated decreasing families
$\{M_k^\mu\}_{k,\mu}$ as defined by equation (\ref{Mkmupi}) in Remark \ref{suitable families}
we have $M^\mu_k(c)\le 1$  and $M^\mu_k(\pi)\le 1$ for $\mu=\overline t_k=\frac1k\sum_{i=1}^kt_i$.

\begin{remark}\label{Rem7}
Let $\{M^\mu_k(c)\}_{k,\mu}$ be decreasing family of the form (\ref{Mkmu}) with test parameters $c_{k-1}=c$ for all $k$.
If for some $c$ we have $M^\mu_k(c)>1$, then 
$\mu<\overline t_k$.
\\
It follows that for an integrated decreasing family $\{M_k^\mu(\pi)\}_{k,\mu}$ 
of the form (\ref{Mkmupi}) also $M_k^\mu(\pi)>1$
implies that $\mu<\overline t_k$.
\\
In particular for decreasing families $\{M_k^\mu(c)\}_{k,\mu}$ and $\{M_k^\mu(\pi)\}_{k,\mu}$
the confidence lower bounds satisfy $B^l_k<\overline t_k$ for all $k$.
\end{remark}
If one decides to stop sampling at time $k$ depending on the combination of 
$\overline t_k$ and $B_k^l$, one should realize that
it is quite possible that the sample mean $\overline t_k$ as an estimator of $\E(T)$
is \emph{biased}.
Of course, if the stopping time $k$ does not depend on the sampling history, 
the sample mean is unbiased estimator of $\E(T)$.
Anyway, the behavior of $(\overline T_k-\E(T))$ for large sample sizes is described by the Law of the Iterated Logarithm. 
\\
\\
For rv's $T$ which are bounded from above, say $T\le1$ a.s. 
$(1-\alpha)$-\emph{confidence upper bounds}
can be constructed as follows:
Given a decreasing family $\{M_k^\mu\}_{k,\mu}$, 
for example modelled after (\ref{Mkmu}) or (\ref{Mkmupi}),
one may consider $\hat M_k^\mu=M_k^{1-\mu}$ applied to a sample $1-T_1,1-T_2,\ldots$ 
of $1-T$.
Then $\{\hat M_k^\mu\}_k$ is a test supermartingale for $H_0:\E(1-T)\le1-\mu$,
equivalent to $H_0:\E(T)\ge\mu$. 
Furthermore $\hat M_k^\mu$ is increasing in $\mu$ for all $k$.
We may refer to the family $\{\hat M_k^\mu\}_{k,\mu}$ as a \emph{increasing} family.
It is convex if $\mu\mapsto \hat M_k^\mu$ is a convex function for all $k$ and all
realizations of the sample of $T$.
Let $B_k^l$ be an $(1-\alpha)$-confidence lower bound for $\E(1-T)$, 
then $B_k^u=1-B_k^l$ is a $(1-\alpha)$-confidence upper bound for $\E(T)$.
It satisfies
$$B_k^u=\infty \hbox{ if } \hat M^\mu_k<1/\alpha\hbox{ for all }\mu>0,\hbox{ and }
B_k^u=\inf\{\mu\mid \hat M^\mu_k\ge1/\alpha\}\hbox{ otherwise. }
$$


\subsubsection*{Confidence intervals}
It is natural to associate confidence \emph{intervals} (rather than regions) for 
$\E(T)$ with a family of
tests of the hypotheses $H_0:\E(T)=\mu$, where $\mu\in\R$.
A desirable property  of the family of tests then is that 
if $H_0:\E(T)=\mu$ cannot be rejected for $\mu=\mu_1$ and $\mu=\mu_2$, 
it cannot be rejected for all $\mu$ between $\mu_1$ and $\mu_2$.
We will produce a family of test supermartingales 
designed to produce confidence intervals for $\E(T)$.
As we will restrict to bounded rv's $T$, we will assume that $0\le T\le1$ a.s..

\begin{theorem}
Suppose $T_1,T_2,...$ is an iid sample of $T$
and let $\{M_k^{\mu,-}\}_{k,\mu}$ and $\{M_k^{\mu,+}\}_{k,\mu}$
be a convex decreasing resp. convex increasing family of test supermartingales
for the null hypotheses $\{H_0:\E(T)\le\mu\}_\mu$, resp. $\{H_0:\E(T)\ge\mu\}_\mu$.
Let $p^-,p^+\ge0$, $p^-+p^+=1$, and $M_k^\mu=p^-M_k^{\mu,-}+p^+M_k^{\mu,+}$.
\\
For all $0<\mu<1$ the process $\{M_k^\mu\}_k$ is a test supermartigale for $H_0:\E(T)=\mu$.
Any realization of 
$\mu\mapsto M_k^\mu$ is a convex function.
Let $0<\alpha<1$.
The region
$R=\{\mu\mid \forall k:M_k^\mu<1/\alpha\}$ is a $(1-\alpha)$-confidence interval.
Thus, for any $k$, $R_k=\{\mu\mid \forall n\le k:M_n^\mu<1/\alpha\}$
as well as $R'_k=\{\mu\mid M_k^\mu<1/\alpha\}$
are $(1-\alpha)$-confidence intervals.
When for the decreasing and increasing family, families of the form 
(\ref{Mkmupi})
with test measures $\pi^-$, resp. $\pi^+$ are used, 
confidence interval $R'_k$ contains the $k$-th sample average $\overline t_k$.
\end{theorem}
\begin{proof}
First of all, suppose $\E(T)=\nu$.
Then $\{M^\nu_n(\pi)\}_n$ is a supermartingale, so 
$\P\{\nu\in R\}=\P\{\forall n:M^\nu_n(\pi)<1/\alpha\}\ge1-\alpha$.
This means that $R$ is a $(1-\alpha)$-confidence region.
Since a convex combination of convex functions is convex, $\mu\mapsto M_n^\mu$ is a convex function.
Thus for each $k$ the set 
$\{\mu\mid M_k^\mu<1/\alpha\}$ is an interval, as well as their intersection
$\{\mu\mid \forall k:M_k^\mu<1/\alpha\}$.
According to Remark \ref{Rem7}, for $\mu=\overline t_k$ we have $M_n^{\mu,+}\le1$
and $M_n^{\mu,-}\le1$ if $\{M_k^{\mu,-}\}_{k,\mu}$ and $\{M_k^{\mu,+}\}_{k,\mu}$ 
are of the form (\ref{Mkmupi}) with test measures $\pi^-$ resp $\pi^+$
so that also $M^\mu_n\le1$.
\end{proof}
\\
\\
Notice that if one set off the investigation of $T$ by looking for a suitable lower bound 
but at some time discovered that also a suitable upper bound is necessary,
it is not safe to switch to the above method and reuse the data!
On the other hand, if one is interested in a $(1-\alpha^-)$-confidence lower bound
using some decreasing family of test supermartingales for the family of null hypotheses
$H_0:\E(T)\le\mu$,
one could as well have maintained at the same time an increasing family of test supermartingales for the null hypotheses
$H_0:\E(T)\ge\mu$, for constructing
$(1-\alpha^+)$-confidence upper bounds. 
We may construct $(1-\alpha^--\alpha^+)$-confidence intervals as follows:
\begin{theorem}
Suppose $T_1,T_2,...$ is an iid sample of $T$
and let $\{M_k^{\mu,-}\}_{k,\mu}$ and $\{M_k^{\mu,+}\}_{k,\mu}$
be a decreasing resp. increasing family of test supermartingales
for the family of null hypotheses $\{H_0:\E(T)\le\mu\}_\mu$, resp. $\{H_0:\E(T)\ge\mu\}_\mu$.
Let $\alpha^+>0$ and $\alpha^->0$.
Let $B_k^l$ be the $(1-\alpha^-)$-confidence lower bounds based on 
$\{M_k^{\mu,-}\}_{k,\mu}$, and $B_k^u$ be the $(1-\alpha^+)$-confidence upper bounds based on 
$\{M_k^{\mu,+}\}_{k,\mu}$.
Then the intervals $(\max_{i\le k}B_i^l,\min_{i\le k}B_i^u)$ are $(1-\alpha^+-\alpha^-)$-confidence intervals.
If one used the decreasing and increasing familiy of the form (\ref{Mkmupi}) it holds that 
$B_k^l< \overline t_k< B_k^u$ for all $k$.
\end{theorem}
It is an unpleasant feature of the above procedures
that it may happen (of course with probability at most $\alpha$, resp. $\alpha^-+\alpha^+$)
that the constructed confidence interval ends up empty.
If one would like to avoid weird conclusions, one could stick to one of the
confidence intervals $R_k$ that is not empty or to $R'_k$ for some $k$.


~\\
\\
Institute for Mathematics, Astrophysics and Particle Physics (IMAPP),
\\
Faculty of Science,
Radboud University Nijmegen,
\\
Heyendaalseweg 135, 6525 AJ Nijmegen, The Netherlands
\\
E-mail: \texttt{H.Hendriks@math.ru.nl}
\end{document}